\documentclass[11pt]{article}

\usepackage{amsmath}
\usepackage{amssymb}
\usepackage{amsthm}
\usepackage{fullpage}
\usepackage{rotating}
\usepackage{enumitem}
\usepackage{bm}
\usepackage{xcolor}

\newtheorem{theorem}{Theorem}
\newtheorem{corollary}[theorem]{Corollary}
\newtheorem{lemma}[theorem]{Lemma}

\newtheorem{definition}[theorem]{Definition}
\newtheorem{claim}[theorem]{Claim}

\newtheorem{remark}[theorem]{Remark}

\newtheorem*{claim*}{Claim}

\newcommand{\prob}[2]{\mathop{\mathrm{Pr}}_{#1}[#2]}
\newcommand{\avg}[2]{\mathop{\textbf{E}}_{#1}\left[#2\right]}
\newcommand{\poly}{\mathop{\mathrm{poly}}}

\newcommand{\abs}[1]{\left|#1\right|}

\newcommand{\Tdt}{\mathrm{T_{\text{DT}}}}

\newcommand{\ip}[2]{\langle #1, #2 \rangle}

\newcommand{\mc}[1]{\mathcal{#1}}

\newcommand{\ACC}{\mathrm{ACC}}
\newcommand{\TC}{\mathrm{TC}}

\newcommand{\sgn}{\mathrm{sgn}}

\newcommand{\coeff}{\mathrm{coeff}}
\usepackage{tikz}

\newcommand{\comment}[2]{\marginpar{\tiny{\textbf{#1 }\textit{#2}}}}
\newcommand{\vaibhav}[1]{\comment{}{}}
\newcommand{\nutan}[1]{\comment{}{}}
\newcommand{\srikanth}[1]{\comment{}{}}
\newcommand{\deepanshu}[1]{\comment{}{}}
\newcommand{\swapnam}[1]{\comment{}{}}

\title{A \#SAT Algorithm for Small Constant-Depth Circuits \\ with PTF gates}
\author{
Swapnam Bajpai
\thanks{Department of Computer Science and Engineering, IIT Bombay, Mumbai, India. \texttt{swapnam@cse.iitb.ac.in}}
\and
Vaibhav Krishan
\thanks{Department of Computer Science and Engineering, IIT Bombay, Mumbai, India. \texttt{vaibhkrishan@iitb.ac.in}}
\and
Deepanshu Kush
\thanks{Department of Mathematics, IIT Bombay, Mumbai, India. \texttt{kush.deepanshu@gmail.com}}
\and
Nutan Limaye
\thanks{Department of Computer Science and Engineering, IIT Bombay, Mumbai, India. \texttt{nutan@cse.iitb.ac.in}}
\and
Srikanth Srinivasan
\thanks{Department of Mathematics, IIT Bombay, Mumbai, India. \texttt{srikanth@math.iitb.ac.in}}
}

\begin{document}

\maketitle

\begin{abstract}
We show that there is a randomized algorithm that, when given a small constant-depth Boolean circuit $C$ made up of gates that compute constant-degree Polynomial Threshold functions or PTFs (i.e., Boolean functions that compute signs of constant-degree polynomials), counts the number of satisfying assignments to $C$ in significantly better than brute-force time.

Formally, for any constants $d,k$, there is an $\varepsilon > 0$ such that the algorithm counts the number of satisfying assignments to a given depth-$d$ circuit $C$ made up of $k$-PTF gates such that $C$ has size at most $n^{1+\varepsilon}$. The algorithm runs in time $2^{n-n^{\Omega(\varepsilon)}}.$

Before our result, no algorithm for beating brute-force search was known even for a single degree-$2$ PTF (which is a depth-$1$ circuit of linear size).

The main new tool is the use of a learning algorithm for learning degree-$1$ PTFs (or Linear Threshold Functions) using comparison queries due to Kane, Lovett, Moran and Zhang (FOCS 2017). We show that their ideas fit nicely into a memoization approach that yields the \#SAT algorithms.
\end{abstract}

\section{Introduction}

This paper adds to the growing line of work on \emph{circuit-analysis algorithms}, where we are given as input a Boolean circuit $C$ from a fixed class $\mathcal{C}$ computing a function $f:\{-1,1\}^n \rightarrow \{-1,1\}$\footnote{We work with the $\{-1,1\}$ basis for Boolean functions, which is by now standard in the literature. (See for instance~\cite{ODonnell}.) Here $-1$ stands for True and $1$ stands for False.} and we are required to compute some parameter of the function $f$. A typical example of this is the question of satisfiability, i.e. whether $f$ is the constant function $1$ or not. In this paper, we are interested in computing \#SAT$(f)$, which is the number of satisfying assignments of $f$ (i.e. $\abs{\{a \in \{-1,1\}^n \mid f(a) = -1\}}$). 

Problems of this form can always be solved by ``brute-force" in time $\poly(|C|)\cdot 2^n$ by trying all assignments to $C$. The question is can this brute-force algorithm be significantly improved, say to time $2^n/n^{\omega(1)}$ when $C$ is small, say $|C| \leq n^{O(1)}$.

Such algorithms, intuitively are able to distinguish a small circuit $C \in \mathcal{C}$ from a ``black-box" and hence find some structure in $C$. This structure, in turn, is useful in answering other questions about $\mathcal{C}$, such as proving lower bounds against the class $\mathcal{C}$.\footnote{This intuition was provided to us by Ryan Williams.} There has been a large body of work in this area, a small sample of which can be found in \cite{PPZ,PPSZ,Will10,Will11}. A striking result of this type was proved by Williams~\cite{Will10} who showed that for many circuit classes $\mathcal{C}$, even \emph{co-non-deterministic} satisfiability algorithms running in better than brute-force time yield lower bounds against $\mathcal{C}$. 

Recently, researchers have also uncovered tight connections between many combinatorial problems and circuit-analysis algorithms, showing that even modest improvements over brute-force search can be used to improve long-standing bounds for these combinatorial problems (see, e.g.,~\cite{VVWsurvey,AHWW,AR,AB}). This yields further impetus in improving known circuit-analysis algorithms. \vaibhav{any references we should add here?}\srikanth{Added.}

This paper is concerned with \#SAT algorithms for constant depth threshold circuits, denoted as $\TC^0$, which are Boolean circuits where each gate computes a linear threshold function (LTF); an LTF computes a Boolean function which accepts or rejects based on the sign of a (real-valued) linear polynomial evaluated on its input. Such circuits are surprisingly powerful: for example, they can perform all integer arithmetic efficiently~\cite{BCH,HAB}, and are at the frontier of our current lower bound techniques~\cite{KW,Chen}. 

It is natural, therefore, to try to come up with circuit-analysis algorithms for threshold circuits. Indeed, there has a large body of work in the area (reviewed below), but some extremely simple questions remain open. 

An example of such a question is the existence of a better-than-brute-force algorithm for satisfiability of degree-$2$ PTFs. Informally, the question is the following: we are given a quadratic polynomial $Q(x_1, \ldots, x_n)$ in $n$ Boolean variables and we ask if there is any Boolean assignment $a \in \{-1,1\}^n$ to $x_1,\ldots, x_n$ such that $Q(a) <0$. (Note that for a linear polynomial instead of a quadratic polynomial, this problem is trivial.)

Surprisingly, no algorithm is known for this problem that is significantly better than $2^n$ time.\footnote{An algorithm for claimed for this problem in the work of Sakai, Seto, Tamaki and Teruyama~\cite{SSTT}. Unfortunately, the proof of this claim is flawed. See Footnote 1 on page 4 of~\cite{KL}.} In this paper, we solve the stronger counting variant of this problem for any constant-degree PTFs. We start with some definitions and then describe this result. 

\begin{definition}[Polynomial Threshold Functions]
\label{def:ptf}
A \emph{Polynomial Threshold Function (PTF) on $n$ variables of degree-$k$} is a Boolean function $f:\{-1,1\}^n\rightarrow \{-1,1\}$ such that there is a degree-$k$ multilinear polynomial $P(x_1,\ldots,x_n)\in \mathbb{R}[x_1,\ldots,x_n]$ that, for all $a\in \{-1,1\}^n,$ satisfies $f(a) = \sgn(P(a)).$ (We assume that $P(a)\neq 0$ for any $a\in \{-1,1\}^n$.)

In such a scenario, we call $f$ a \emph{$k$-PTF}. In the special case that $k = 1$, we call $f$ a \emph{Linear Threshold function (LTF)}. We also say that the polynomial $P$ \emph{sign-represents} $f$.

We define weight of $P$, denoted as $w(P)$, to be the bit-complexity of the sum of absolute values of all the coefficients of $P$. 
\end{definition}

The \#SAT problem for $k$-PTFs is the problem of counting the number of assignments that satisfy a given $k$-PTF $f$. Formally,
\begin{definition}[\#SAT problem for $k$-PTFs]
The problem is defined as follows. 
\begin{itemize}
\item[]\textbf{Input}: A $k$-PTF $f$, specified by a degree-$k$ polynomial $P(x_1,\ldots,x_n)$ with integer coefficients.\footnote{It is known~\cite{Muroga} that such a representation always exists.}
\item[]\textbf{Output}: The number of satisfying assignments to $f$. That is, the number of $a\in \{-1,1\}^n$ such that $P(a) < 0.$
\end{itemize}
We use \#SAT$(f)$ to denote this output. We say that the input instance has parameters $(n,M)$ if $n$ is the number of input variables and $w(P)\leq M$. 
\end{definition}
\begin{remark}
\label{rem:M}
An interesting setting of $M$ is $\poly(n)$ since any $k$-PTF can be represented by an integer polynomial with coefficients of bit-complexity at most $\tilde{O}(n^k)$~\cite{Muroga}.
\end{remark}

We give a better-than-brute-force algorithm for \#SAT$(k\text{-PTF})$. Formally we prove the following theorem. 

\begin{theorem}
\label{thm:ptfsat-intro}
Fix any constant $k$. There is a zero-error randomized algorithm that solves the \#SAT problem for $k$-PTFs in time $\poly(n,M)\cdot2^{n-S}$ where $S = \tilde{\Omega}(n^{1/(k+1)})$ and $(n,M)$ are the parameters of the input $k$-PTF $f.$ (The $\tilde{\Omega}(\cdot)$ hides factors that are inverse polylogarithmic in $n$.) 
\end{theorem}

We then extend this result to a powerful model of circuits called \emph{$k$-PTF circuits}, where each gate computes a $k$-PTF. This model was first studied by Kane, Kabanets and Lu~\cite{KKL} who proved strong average case lower bounds for slightly superlinear-size constant-depth $k$-PTF circuits. Using these ideas, Kabanets and Lu~\cite{KL} were able to give a \#SAT algorithm for a restricted class of  $k$-PTF circuits, where each gate computes a PTF with a subquadratically many, say $n^{1.99}$, monomials (while the size remains the same, i.e. slightly superlinear).\footnote{Their result also works for the slightly larger class of PTFs that are subquadratically sparse in the $\{0,1\}$-basis with \emph{no restriction} on degree. Our result can also be stated for the larger class of \emph{polynomially sparse} PTFs, but for the sake of simplicity, we stick to constant-degree PTFs.} A reason for this restriction on the PTFs was that they did not have an algorithm to handle even a single degree-$2$ PTF (which can have $\Omega(n^2)$ many monomials).

Building on our \#SAT algorithm for $k$-PTFs and the ideas of~\cite{KL}, we are able to handle general $k$-PTF circuits of slightly superlinear size. We state these results formally below. 

We first define $k$-PTF circuits formally.
\begin{definition}[$k$-PTF circuits]
\label{def:ptf-ckt}
A \emph{$k$-PTF circuit} on $n$ variables is a Boolean circuit on $n$ variables where each gate $g$ of fan-in $m$ computes a fixed $k$-PTF of its $m$ inputs. The size of the circuit is the \emph{number of wires} in the circuit, and the depth of the circuit is the longest path from an input to the output gate.\footnote{Note, crucially, that only the fan-in of a gate counts towards its size. So any gate computing a $k$-PTF on $m$ variables only adds $m$ to the size of the circuit, though of course the polynomial representing this PTF may have $\approx m^k$ monomials.}
\end{definition}

The problems we consider is the \#SAT problem for $k$-PTF circuits, defined as follows.
\begin{definition}[\#SAT problem for $k$-PTF circuits]
The problem is defined as follows. 
\begin{itemize}
\item[]\textbf{Input}: A $k$-PTF circuit $C$, where each gate $g$ is labelled by an integer polynomial that sign-represents the function that is computed by $g$.
\item[]\textbf{Output}: The number of satisfying assignments to $C$.
\end{itemize}
We use \#SAT$(C)$ to denote this output. We say that the input instance has parameters $(n,s,d,M)$ where $n$ is the number of input variables, $s$ is the size of $C$, $d$ is the depth of $C$ and $M$ is the maximum over the weights of the degree-$k$ polynomials specifying the $k$-PTFs in $C$. We will say that $M$ is the weight of $C$, denoted by $w(C)$. 
\end{definition}

We now state our result on \#SAT for $k$-PTF circuits. The following result implies Theorem~\ref{thm:ptfsat-intro}, but we prove them separately. 

\begin{theorem}
\label{thm:ptfcktsat-intro}
Fix any constants $k,d$. Then the following holds for some constant $\varepsilon_{k,d} > 0$ depending on $k,d.$ There is a zero-error randomized algorithm that solves the \#SAT problem for $k$-PTF circuits of size at most $s = n^{1+\varepsilon_{k,d}}$  with probability at least $1/4$ and outputs $?$ otherwise. The algorithm runs in time $\poly(n,M)\cdot 2^{n-S}$, where $S = n^{\varepsilon_{k,d}}$ and $(n,s,d,M)$ are the parameters of the input $k$-PTF circuit. 
\end{theorem}

\paragraph{Previous work.} Satisfiability algorithms for $\TC^0$ have been widely investigated.  Impagliazo, Lovett, Paturi and Saks~\cite{IPS,ILPS} give algorithms for checking satisfiability of depth-$2$ threshold circuits with $O(n)$ gates. An incomparable result was proved by Williams~\cite{Will14} who obtained algorithms for subexponential-sized circuits from the class $\ACC^0 \circ \text{LTF}$, which is a subclass of subexponential $\TC^0$.\footnote{$\ACC^0 \circ \text{LTF}$ is a subclass of $\TC^0$ where general threshold gates are allowed only just above the variables. All computations above these gates are one of AND, OR or Modular gates (that count the number of inputs modulo a constant). It is suspected (but not proved) that subexponential-sized $\ACC^0$ circuits cannot simulate even a single general threshold gate. Hence, it is not clear if the class of subexponential-sized $\ACC^0\circ \text{LTF}$ circuits contains even depth-$2$ $\TC^0$ circuits of linear size.} For the special case of $k$-PTFs (and generalizations to sparse PTFs over the $\{0,1\}$ basis) with \emph{small weights}, a \#SAT algorithm follows from the result of Sakai et al.~\cite{SSTT}.

For general constant-depth threshold circuits, the first satisfiability algorithm was given by Chen, Santhanam and Srinivasan~\cite{CSS}. In their paper, Chen et al. gave the first average case lower bound for $\TC^0$ circuits of slightly super linear size $n^{1+\varepsilon_d}$, where $\varepsilon_d$ depends on the depth of the circuit. (These are roughly the strongest size lower bounds we know for general $\TC^0$ circuits even in the worst case~\cite{IPS97}.) Using their ideas, they gave the first (zero-error randomized) improvement to brute-force-search for satisfiability algorithms (and indeed even \#SAT algorithms) for constant depth $\TC^0$ circuits of size at most $n^{1+\varepsilon_d}$.  

The lower bound results of~\cite{CSS} were extended to the much more powerful class of $k$-PTF circuits (of roughly the same size as~\cite{CSS}) by Kane, Kabanets and Lu~\cite{KKL}. In a follow-up paper, Kabanets and Lu~\cite{KL} considered the satisfiability question for $k$-PTF circuits, and could resolve this question in the special case that each PTF is subquadratically sparse, i.e. has $n^{2-\Omega(1)}$ monomials. One of the reasons for this sparsity restriction is that their strategy does not seem to yield a SAT algorithm for a single degree-$2$ PTF (which is a depth-$1$ $2$-PTF circuit of \emph{linear} size). 

\subsection{Proof outline.} For simplicity we discuss SAT algorithms instead of \#SAT algorithms. 

\subsubsection*{Satisfiability algorithm  for $k$-PTFs.} At a high-level, the algorithm uses Memoization, which is a standard and very useful strategy for satisfiability algorithms (see, e.g. \cite{Santhanam}). Let $\mathcal{C}$ be a circuit class and $\mathcal{C}_n$ be the subclass of circuits from $\mathcal{C}$ that have $n$ variables. Memoization algorithms for $\mathcal{C}$-SAT fit into the following two-step template.

\begin{itemize}
\item Step 1: Solve by brute-force all instances of $\mathcal{C}$-SAT where the input circuit $C' \in \mathcal{C}_m$ for some suitable $m \ll n$. (Typically, $m = n^\varepsilon$ for some constant $\varepsilon$.) Usually this takes $\exp(m^{O(1)}) \ll 2^n$ time. 

\item Step 2: On the input $C \in \mathcal{C}_n$, set all input variables $x_{m+1}, \ldots, x_n$ to Boolean values and for each such setting, obtain $C'' \in \mathcal{C}_m$ on $m$ variables. Typically $C''$ is a circuit for which we have solved satisfiability in Step 1 and hence by a simple table lookup, we should be able to check if $C''$ is satisfiable in $\poly(|C|)$ time. Overall, this takes time $O^*(2^{n-m}) \ll 2^n$.
\end{itemize}

At first sight, this seems perfect for $k$-PTFs, since it is a standard result that the number of $k$-PTFs on $m$ variables is at most $2^{O(m^{k+1})}$~\cite{chow}. Thus, Step 1 can be done in $2^{O(m^{k+1})} \ll 2^n$ time. 

For implementing Step 2, we need to ensure that the lookup (for satisfiability for $k$-PTFs on $m$ variables) can be done quickly. Unfortunately how to do this is unclear. The following two ways suggest themselves. 

\begin{itemize}
\item Store all polynomials $P' \in \mathbb{Z}[x_1, \ldots, x_m]$ with small coefficients. Since every $k$-PTF $f$ can be sign-represented by an integer polynomial with coefficients of size $2^{\poly(m)}$~\cite{Muroga}, this can be done with a table of size $2^{\poly(m)}$ and in time $2^{\poly(m)}$. When the coefficients are small (say of bit-complexity $\leq n^{o(1)}$), then this strategy already yields a \#SAT algorithm, as observed by Sakai et al.~\cite{SSTT}. Unfortunately, in general, given a restriction $P'' \in \mathbb{Z}[x_1, \ldots, x_m]$ of a polynomial $P \in \mathbb{Z}[x_1, \ldots, x_n],$ its coefficients can be much larger (say $2^{\poly(n)}$) and it is not clear how to efficiently find a polynomial with small coefficients that sign-represents the same function. 

\item It is also known that every $k$-PTF on $m$ variables can be uniquely identified by $\poly(m)$ numbers of bit-complexity $O(m)$ each~\cite{chow}: \nutan{This statement is sounding a bit strange.} these are called the ``Chow parameters'' of $f$. Again for this representation, it is unclear how to compute efficiently the Chow parameters of the function represented by the restricted polynomial $P''$. (Even for an LTF, computing the Chow parameters is as hard as Subset-sum~\cite{OS}.)
\end{itemize}

The way we solve this problem is by using a beautiful recent result of Kane, Lovett, Moran and Zhang~\cite{KLMZ}, who show that there is a simple decision tree that, when given as input the coefficients of any degree-$k$ polynomial $P' \in \mathbb{Z}[x_1, \ldots, x_m]$, can determine the sign of the polynomial $P'$ at all points in $\{-1,1\}^m$ using only $\poly(m)$ queries to the coefficients of $P$. Here, each query is a linear inequality on the coefficients of $P$; such a decision tree is called a \emph{linear decision tree}. 

Our strategy is to replace Step 1 with the construction of this linear decision tree (which can be done in $\exp(m^{O(1)})$ time). At each leaf of the linear decision tree, we replace the truth table of the input polynomial $P'$ by a single bit that indicates whether $f' = \text{sgn}(P')$ is satisfiable or not. 

In Step 2, we simply run this decision tree on our restricted polynomial $P''$ and obtain the answer to the corresponding satisfiability query in $\poly(m,w(P''))$ time. Note, crucially, that the height of the linear decision tree implied by~\cite{KLMZ} construction is $\poly(m)$ and \emph{independent} of the bit-complexity of the coefficients of the polynomial $P''$ (which may be as big as $\poly(n)$
 in our algorithm). This concludes the description of the algorithm for $k$-PTF. 

\subsubsection*{Satisfiability algorithm for $k$-PTF circuits.} For $k$-PTF circuits, we follow a template set up by the result of Kabanets and Lu~\cite{KL} on sparse-PTF circuits. We start by describing this template and then describe what is new in our algorithm. 

The Kabanets-Lu algorithm is an induction on the depth $d$ of the circuit (which is a fixed constant). Given as input a depth $d$ $k$-PTF circuit $C$ on $n$ variables, Kabanets and Lu do the following:

{Depth-reduction:} In \cite{KL}, it is shown that on a random restriction that sets all \emph{but} $n^{1-2\beta}$ variables (here, think of $\beta$ as a small constant, say $0.01$) to random Boolean values, the bottom layer of $C$ simplifies in the following sense. 

All but $t \leq n^\beta$ gates at the bottom layer become \emph{exponentially} biased, i.e. on all but $\delta = \exp(-n^{\Omega(1)})$ fraction of inputs they are equal to a fixed $b \in \{-1,1\}$. Now, for each such biased gate $g,$ there is a minority value $b_g \in \{-1,1\}$ that it takes on very few inputs. \cite{KL} show how to enumerate this small number of inputs in $\delta \cdot 2^n$ time and check if there is a satisfying assignment among these inputs. Having ascertained that there is no such assignment, we replace these gates by their majority value and there are only $t$ gates at the bottom layer. At this point, we ``guess" the output of these $t$ ``unbiased" gates and for each such  guess $\sigma \in \{-1,1\}^t$, we check if there is an assignment that simultaneously satisfies:
\begin{enumerate}
    \item[(a)] the depth $d-1$ circuit $C'$, obtained by setting the unbiased gates to the guess $\sigma$, is satisfied.
    \item[(b)] each unbiased gate $g_i$ evaluates to the corresponding value $\sigma_i$.
\end{enumerate}

{Base case:} Continuing this way, we eventually get to a base case which is an AND of sparse PTFs for which there is a satisfiability algorithm using the polynomial method. 

In the above algorithm, there are two steps where subquadratic sparsity is crucially used. The first is the minority assignment enumeration algorithm for PTFs, which uses ideas of Chen and Santhanam~\cite{CS} to reduce the problem to enumerating biased LTFs, which is easy~\cite{CSS}. The second is the base case, which uses a non-trivial polynomial approximation for LTFs~\cite{Sri13}. Neither of these results hold for even degree-$2$ PTFs in general. To overcome this, we do the following.

\paragraph{Enumerating minority assignments.} Given a $k$-PTF on $m$ variables that is $\delta=\exp(-n^{\Omega(1)})$-close to $b \in \{-1,1\}$, we enumerate its minority assignments as follows. First, we set up a linear decision tree as in the $k$-PTF satisfiability algorithm. Then we set all but $q \approx \log \frac{1}{\delta}$ variables of the PTF. On most such settings, the resulting PTF becomes the constant function and we can check this using the linear decision tree we created earlier. In this setting, there is nothing to do. Otherwise, we brute-force over the remaining variables to find the minority assignments. Setting parameters suitably, this yields an $O(\sqrt{\delta} \cdot 2^m)$ time algorithm to find the minority assignments of a $\delta$-biased $k$-PTF on $m$ variables. 

\paragraph{Base case:} Here, we make the additional observation (which~\cite{KL} do not need) that the AND of PTFs that is obtained further is \emph{small} in that it only has slightly superlinear size. Hence, we can apply another random restriction in the style of~\cite{KL} and using the minority assignment enumeration ideas, reduce it to an AND of a small (say $n^{0.1}$) number of PTFs on $n^{0.01}$ (say) variables. At this point, we can again run the linear decision tree (in a slightly more generalized form) to check satisfiability.

\section{A result of Kane, Lovett, Moran, and Zhang~\cite{KLMZ}}

\begin{definition}[Coefficient vectors.] Fix any $k,m\geq 1.$ There are exactly $r = \sum_{i=0}^k \binom{m}{i}$ many multilinear monomials of degree at most $k$. Any multilinear polynomial $P(x_1,\ldots,x_m)$ can be identified with a list of the coefficients of its monomials in lexicographic order (say) and hence with some vector $w\in \mathbb{R}^r$. We call $w$ the \emph{coefficient vector} of $P$ and use $\coeff_{m,k}(P)$ to denote this vector. When $m,k$ are clear from context, we will simply use $\coeff(P)$ instead of $\coeff_{m,k}(P).$
\end{definition}

\begin{definition}[Linear Decision Trees]
A \emph{Linear Decision Tree} for a function $f:\mathbb{R}^r \rightarrow S$ (for some set $S$) is a  decision tree where each internal node is labelled by a linear inequality of the form $\sum_{i=1}^r w_i z_i \geq \theta$ (here $z_1,\ldots,z_n$ denote the input variables). Depending on the answer to this linear inequality, computation proceeds to the left or right child of this node, and this process continues until a leaf is reached, which is labelled with an element of $S$ that is the output of $f$ on the given input.
\end{definition}

The following construction of linear decision trees due to Kane, Lovett, Moran and Zhang~\cite{KLMZ} will be crucial for us.

\begin{theorem}
\label{thm:KLM}
There is a randomized algorithm, which on input a positive integer $r$, a subset $H\subseteq \{-1,1\}^r$, and an error parameter $\varepsilon,$ produces a (random) linear decision tree $T$ of depth $\Delta = O(r\log r\cdot \log(|H|/\varepsilon))$ that computes a (random) function $F:\mathbb{R}^r\rightarrow \{-1,1\}^{|H|}\cup \{?\}$ that has the following properties.
\begin{enumerate}
\item Each linear query has coefficients in $\{-2,-1,0,1,2\}.$
\item Given as input any $w\in \mathbb{R}^r$ such that $\ip{w}{a} \neq 0$ for all $a\in \{-1,1\}^r$, $F(w)$ is either the truth table of the LTF defined by $w$ (with constant term $0$) on inputs from $H\subseteq \{-1,1\}^r$, or is equal to $?$. Further, we have $\prob{F}{F(w) = ?} \leq \varepsilon.$
\end{enumerate}
The randomized algorithm runs in time $2^{O(\Delta)}.$
\end{theorem}

\begin{remark}
\label{rem:algo-KLMZ}
The last statement in the above theorem is not formally stated in~\cite{KLMZ} but can easily be inferred from their proof and a remark~\cite[Page 363]{KLMZ} on the ``Computational Complexity'' of their procedure.\footnote{We also thank Daniel Kane (personal communication) for telling us about this.} 
\end{remark}

We will need a generalization of this theorem for evaluating (tuples of) $k$-PTFs. However, it is a simple corollary of this theorem.

\begin{corollary}
\label{cor:KLM}
Fix positive constants $k$ and $c$. Let $r = \sum_{i=0}^k\binom{m}{i} = \Theta(m^k)$ denote the number of coefficients in a degree-$k$ multilinear polynomial in $m$ variables. There is a randomized algorithm which on input positive integers $m$ and $\ell\leq m^c$ produces a (random) linear decision tree $T$ of depth $\Delta = O(\ell\cdot m^{k+1}\log m)$ that computes a (random) function $F:\mathbb{R}^{r\cdot \ell}\rightarrow \mathbb{N}\cup \{?\}$ that has the following properties.
\begin{enumerate}
\item Each linear query has coefficients in $\{-2,-1,0,1,2\}.$
\item Given as input any $\ell$-tuple of coefficient vectors $\overline{w} = (\coeff_{m,k}(P_1),\ldots,\coeff_{m,k}(P_\ell))\in \mathbb{R}^{r\cdot\ell}$ such that $P_i(a) \neq 0$ for all $a\in \{-1,1\}^m$, $F(\overline{w})$ is either the number of common satisfying assignments to all the $k$-PTFs on $\{-1,1\}^m$ sign-represented by $P_1,\ldots,P_\ell$, or is equal to $?$. Further, we have $\prob{F}{F(\overline{w}) = ?} \leq (1/2).$
\end{enumerate}
The randomized algorithm runs in time $2^{O(\Delta)}.$
\end{corollary}

\begin{proof}
For each $b\in \{-1,1\}^m,$ define $\mathrm{eval}_b\in \{-1,1\}^r$ to be the vector of all evaluations of multilinear monomials of degree at most $k$, taken in lexicographic order, on the input $b$. Define $H\subseteq \{-1,1\}^r$ to be the set $\{\mathrm{eval}_b\ |\ b\in \{-1,1\}^m\}.$ Clearly, $|H| \leq 2^m$.
Further, note that given any polynomial $P(x_1,\ldots,x_m)$ of degree at most $k$, the truth table of the $k$-PTF sign-represented by $P$ is given by the evaluation of the LTF represented by $\coeff(P)$ at the points in $H$. Our aim, therefore, is to evaluate the LTFs corresponding to $\coeff(P_1),\ldots,\coeff(P_\ell)$ at all the points in $H$.

For each $i$, we use the randomized algorithm from Theorem~\ref{thm:KLM} to produce a decision tree $T_i$ that evaluates the Boolean function $f_i:\{-1,1\}^m\rightarrow\{-1,1\}$ sign-represented by $P_i$ (or equivalently, evaluating the LTF corresponding to $\coeff(P_i)$ at all points in $H$) with error  $\varepsilon = 1/2\ell.$ Note that $T_i$ has depth $O(m^k\log m\cdot \log(2^m/\ell)) = O(m^{k+1}\log m)$ as $\ell \leq m^c$. The final tree $T$ is obtained by simply running $T_1,\ldots,T_\ell$ in order, which is of depth $O(\ell m^{k+1}\log m).$ The tree $T$ outputs the number of common satisfying assignments to all the $f_i$ if all the $T_i$s succeed and $?$ otherwise. Since each $T_i$ outputs $?$ with probability at most $1/2\ell$, the tree $T$ outputs $?$ with probability at most $(1/2\ell)\cdot \ell  = 1/2$.

The claim about the running time follows from the analogous claim in Theorem~\ref{thm:KLM} and the fact that the number of common satisfying assignments to all the $f_i$ can be computed from the truth tables in $2^{O(m)}$ time. This completes the proof.
\end{proof}

\section{The PTF-SAT algorithm}

We are now ready to prove Theorem~\ref{thm:ptfsat-intro}. We first state the algorithm, which follows a standard memoization idea (see, e.g. \cite{Santhanam}). We assume that the input is a polynomial $P\in \mathbb{Z}[x_1,\ldots,x_n]$ of degree at most $k$ that sign-represents a Boolean function $f$ on $n$ variables. The parameters of the instance are assumed to be $(n,M)$. Set $m = n^{1/(k+1)}/\log n$.

\noindent
Algorithm $\mathbf{\mc{A}}$
\begin{enumerate}
\item\label{tree_prod} Use $n_1 = 10n$ \emph{independent runs} of the algorithm from Corollary~\ref{cor:KLM} with $\ell=1$ to construct independent random linear decision trees $T_1,\ldots,T_{n_1}$ such that on any input polynomial $Q(x_1,\ldots,x_m)$ (or more precisely $\coeff_{m,k}(Q)$) of degree at most $k$ that sign-represents an $k$-PTF $g$ on $m$ variables, each $T_i$ computes the number of satisfying assignments to $g$ with error at most $1/2$.
\item Set $N=0$. ($N$ will ultimately be the number of satisfying assignments to $f$.)
\item For each setting $\sigma\in \{-1,1\}^{n-m}$ to the variables $x_{m+1},\ldots,x_n$, do the following:
\begin{enumerate}
\item Compute the polynomial $P_\sigma$ obtained by substituting the variables $x_{m+1,\ldots,x_n}$ accordingly in $P$. 
\item\label{alg_output} Run the decision trees $T_1,\ldots,T_{n_1}$ on $\coeff(P_\sigma)$ and compute their outputs. If all the outputs are $?$, output $?$. Otherwise, some $T_i$ outputs $N_\sigma$, the number of satisfying assignments to $P_\sigma.$ Add this to the current estimate to $N$. 
\end{enumerate}
\item Output $N$.
\end{enumerate}

\paragraph{Correctness.} It is clear from Corollary~\ref{cor:KLM} and step~\ref{alg_output} that algorithm $\mc{A}$ outputs either $?$ or the correct number of satisfying assignments to $f$.  Further, we claim that with probability at least $1-1/2^{\Omega(n)},$ the output is indeed the number of satisfying assignments to $f$. To see this, observe that it follows from Corollary~\ref{cor:KLM} that for each setting $\sigma\in \{-1,1\}^{n-m}$ to the variables $x_{m+1},\ldots,x_n$, each linear decision tree $T_i$ produced in step~\ref{tree_prod} errs on $\coeff(P_\sigma)$ (i.e. outputs $?$) with probability at most $1/2$. The probability of each $T_i$ doing so is thus at most $1/2^{n_1}$, as they are constructed independently. So the probability that the algorithm fails to determine $N_\sigma$ is at most $1/2^{n_1}$. Finally, taking a union bound over all $\sigma$, which are $2^{n-m}$ in number, we conclude that the probability of algorithm $\mc{A}$ outputting $?$ is at most $2^{n-m}/2^{n_1}\leq 1/2^{\Omega(n)}$.

\paragraph{Running time.} We show that the running time of algorithm $\mc{A}$ is $\poly(n,M)\cdot 2^{n-m}$. First note that by Corollary~\ref{cor:KLM}, the construction of a single linear decision tree $T_i$ takes $2^{O(\Gamma)}$ time, where $\Gamma = m^{k+1}\log m$, and hence, step 1 takes $n_1\cdot 2^{O(\Gamma)}$ time. Next, for a setting $\sigma\in \{-1,1\}^{n-m}$ to the variables $x_{m+1},\ldots,x_n$, computing $P_\sigma$ and constructing the vector $\coeff(P_\sigma)$ takes only $\poly(n,M)$ time. Recall that the depth of each linear decision tree $T_i$ is $O(\Gamma)$ and thus, on input vector $\coeff(P_\sigma)$, each of whose entries has bit complexity at most $M$, it takes time $O(\Gamma)\cdot \poly(M,n)$ to run all $T_i$ and obtain the output $N_\sigma$ or $?$. Therefore, step 3 takes $\poly(n,M)\cdot 2^{n-m}$ time. Finally, the claim about the total running time of algorithm $\mc{A}$ follows at once when we observe that for the setting $m = n^{1/(k+1)}/\log n$, $\Gamma = o(n/ (\log n)^k) = o(n)$.

\section{Constant-depth circuits with PTF gates}
\label{sec:const-d}

In this section we give an algorithm for counting the number of satisfying assignment for a $k$-PTF circuit of constant depth and slightly superlinear size.  We begin with some definitions. 
\begin{definition}
\label{def:delta-close}
Let $\delta \leq 1$ be any parameter. Two Boolean functions $f,g$ are said to be $\delta$-close if $\prob{x}{f(x) \neq g(x)} \leq \delta$. 

A $k$-PTF $f$ specified by a polynomial $P$ is said to be $\delta$-close to an explicit constant if it is $\delta$-close to a constant and such a constant can be computed efficiently, i.e. $\poly(n,M)$, where $n$ is the number of variables in $P$ and $w(P)$ is at most $M$.

\end{definition}

\begin{definition}
\label{def:maj-min}
For a Boolean function $f:\{-1,1\}^n \rightarrow \{-1,1\}$, the majority value of $f$ is the bit value $b \in \{-1,1\}$ which maximizes $\prob{x}{f(x)=b}$ and the bit value $-b$ is said to be its minority value. 

For a Boolean function $f$ with majority value $b$, an assignment $x \in \{-1,1\}^n$ is said to be a majority assignment if $f(x)=b$ and minority assignment otherwise. 
\end{definition}

\begin{definition}
\label{def:Csigma}
Given a $k$-PTF $f$ on $n$ variables specified by a polynomial $P$, a parameter $m \leq n$ and a partial assignment $\sigma \in \{-1,1\}^{n-m}$ on $n-m$ variables, let $P_\sigma$ be the polynomial obtained by substituting the variables in $P$ according to $\sigma$. If $P$ has parameters $(n,M)$ then $P_\sigma$ has parameters $(m,M)$. For a $k$-PTF circuit $C$, $C_\sigma$ is defined similarly.  If $C$ has parameters $(n,s,d,M)$ then $C_\sigma$ has parameters $(m,s,d,M)$. 
\end{definition}

\paragraph{Outline of the \#SAT procedure.} For designing a \#SAT algorithm for $k$-PTF circuits, we use the genric framework developed by Kabanets and Lu~\cite{KL} with some crucial modifications. 

Given a $k$-PTF circuit $C$ on $n$ variables of depth $d$ we want to count the number of satisfying assignments $a \in \{-1,1\}^n$ such that $C(a)=-1$. We in fact solve a slightly more general problem. Given $(C,\mathcal{P})$, where $C$ is a small $k$-PTF circuit of depth $d$ and $\mathcal{P}$ is a set of $k$-PTF functions, such that $\sum_{f \in \mathcal{P}} \text{fan-in}(f)$ is small, we count the number of assignments that simultaneously satisfy $C$ and all the function in $\mathcal{P}$. 

At the core of the algorithm that solves this problem, Algorithm~$\mathcal{B}$, is a recursive procedure $\mathcal{A}_5$, which works as follows: on inputs $(C,\mathcal{P})$ it first applies a simplification step that outputs $\ll 2^n$ instances of the form $(C', \mathcal{P}')$ such that

\begin{itemize}
\item Both $C'$ and functions in $\mathcal{P}'$ are on $m \ll n$ variables. 
\item The sets of satisfying assignments of these instances ``almost'' partition the set of satisfying assignments of $(C,\mathcal{P})$. 
\item With all but very small probability the bottom layer of $C'$ has the following nice structure.
\begin{itemize}
\item At most $n$ gates are $\delta$-biased. We denote this set of gates by $B$ (as we will simplify them by setting them to the values they are biased towards).
\item At most $n^{\beta_d}$ gates are not $\delta$-biased. We denote these gates by $G$ (as we will simplify them by ``guessing" their values). 
\end{itemize}
\item There is a small set of satisfying assignments that are not covered by the satisfying assignments of $(C',\mc{P}')$ but we can count these assignments with a brute-force algorithm that does not take too much time.
\end{itemize}

For each $C'$ with this nice structure, then we try to use this structure to create $C''$ which has depth $d-1$. Suppose we reduce the depth as follows:

\begin{itemize}
\item Set all the gates in $B$ to the values that they are biased towards. 
\item Try all the settings of the values that the gates in $G$ can take, thereby from $C'$ creating possibly $2^{n^{\beta_d}}$ instances $(C'',\mathcal{P}')$. 
\end{itemize}
$(C'',\mathcal{P}')$ now is an instance where $C''$ has depth $d-1$. Unfortunately, by simply setting biased gates to the values they are biased towards, we may miss out on the minority assignments to these gates which could  eventually satisfy $C'$. We design a subroutine $\mathcal{A}_3$ to precisely handle this issue, i.e. to keep track of the number of minority assignments, say $N_{C'}$. This part of our algorithm is completely different from that of~\cite{KL}, which only works for subquadratically sparse PTFs. 

Once $\mathcal{A}_3$ has computed $N_{C'}$, i.e. the number of satisfying assignments among the minority assignments, we now need to only count the number of satisfying assignments among the rest of the assignments. 

To achieve this we use an idea similar to that in~\cite{CSS,KL}, which involves appending $\mathcal{P}'$ with a few more $k$-PTFs (this forces the biased gates to their majority values). This gives say a set $\tilde{\mathcal{P}}'$.  Similarly, while setting gates in $G$ to their guessed values, we again use the same idea to ensure that we are counting satisfying assignments consistent with the guessed values, once again updating $\tilde{\mathcal{P}}'$ to a new set $\mathcal{P}''$. This creates instances of the form $(C'',\mathcal{P}'')$, where the depth of $C''$ is $d-1$. 

This way, we iteratively decrease the depth of the circuit by $1$. Finally, we have instances $(C'',\mathcal{P}'')$ such that the depth of $C''$ is $1$, i.e. it is a single $k$-PTF, say $h$. At this stage we solve \#SAT$(\tilde{C})$, where $\tilde{C} = h \wedge \bigwedge_{f \in \mathcal{P}''} f$. This is handled in a subroutine $\mathcal{A}_4$. Here too our algorithm differs significantly from~\cite{KL}.

\vspace*{10pt}
In what follows we will prove Theorem~\ref{thm:ptfcktsat-intro}. In order to do so, we will set up various subroutines $\mc{A}_1,\mc{A}_2,\mc{A}_3,\mc{A}_4,\mc{A}_5$ designed to accomplish certain tasks and combine them together at the end to finally design algorithm $\mc{B}$ for the \#SAT problem for $k$-PTF circuits.

\(\mc{A}_1\) will be an oracle, used in other routines, which will compute number of common satisfying assignments for small AND of PTFs on few variables (using the same idea as in the algorithm for \#SAT for $k$-PTFs). \(\mc{A}_2\) will be a simplification step, which will allow us to argue to argue about some structure in the circuit (this algorithm is from~\cite{KL}). It will make many gates at the bottom of the circuit \(\delta\)-close to a constant, thus simplifying it. \(\mc{A}_3\) will be used to count minority satisfying assignments for a bunch of \(\delta\)-biased PTFs, i.e. assignments which cause at least one of the PTFs to evaluate to its minority value. \(\mc{A}_4\) will be a general base of case of our algorithm, which will count satisfying assignments for AND of superlinear many PTFs, by first using \(\mc{A}_2\) to simplify the circuit, then reducing it to the case of small AND of PTFs and then using \(\mc{A}_1\). \(\mc{A}_5\) will be a recursive procedure, which will use \(\mc{A}_2\) to first simplify the circuit, and then convert it into a circuit of lower depth, finally making a recursive call on the simplified circuit.

\paragraph{Parameter setting.} Let $d$ be a constant. Let $A,B$ be two fixed absolute large constants. Let \(\zeta = \min(1, A/2Bk^2)\). For each $2 \leq i \leq d$, let $\beta_{i} = A\cdot \varepsilon_{i}$ and $\varepsilon_i = (\frac{\zeta}{10A(k+1)})^i$. Choose $\beta_{1} = 1/10$ and $\varepsilon_{1}=1/10A$. \vaibhav{I think \(\varepsilon_i = \frac{\zeta^{i-1}}{(k+1) (10 A)^i}\) works too, but please verify.}

\paragraph{Oracle access to a subroutine:} Let $\mathcal{A}_1(n',s,f_1,\ldots,f_s)$ denote a subroutine with the following specification. Here, $n$ is the number of variables in the original input circuit.

\begin{itemize}
\item[] \textbf{Input:} AND of $k$-PTFs, say $f_1, \ldots, f_s$ specified by polynomials $P_1, \ldots, P_s$ respectively, such that $s \leq n^{0.1}$ and for each $i \in [s]$, $f_i$ is defined over $n'\leq n^{1/(2(k+1))}$ variables and $w(P_i)\leq M$. 
\item[]  \textbf{Output:} \#$\{a \in \{-1,1\}^{n'} \mid \forall i \in [s], P_i(a) = -1\}$.

\end{itemize}
In what follows, we will assume that we have access to the above subroutine $\mathcal{A}_1$. We will set up such an oracle and show that it answers any call to it in time $\poly(n,M)$ in Section~\ref{sec:put-together}. 

\subsection{Simplification of a $k$-PTF circuit}
\label{sec:depth-reduction}

For any $1 > \varepsilon \gg (\log n)^{-1}$, let $\beta =A \varepsilon $ and $\delta = \exp(-n^{\beta/B\cdot k^2})$, where $A$ and $B$ are constants. Note that it is these constants $A,B$ we use in the parameter settings paragraph above.  
Let $\mathcal{A}_2(C,d,n,M)$ be the following subroutine. 

\begin{itemize}
\item[] {\bf Input:} $k$-PTF circuit $C$ of depth $d$ on $n$ variables with size $n^{1+\varepsilon}$ and weight $M$. 

\item[] {\bf Output:} A decision tree $\Tdt$ of depth $n-n^{1-2\beta}$ such that for a uniformly random leaf $\sigma \in \{-1,1\}^{n-n^{1-2\beta}}$ it outputs a \emph{good} circuit $C_\sigma$ with probability $1-\exp(-n^{\varepsilon})$, where $C_\sigma$ is called good if its bottom layer has the following structure:
\begin{itemize}
\item  there are at most $n$ gates which are $\delta$-close to an explicit constant. Let $B_\sigma$ denote this set of gates. 
\item there are at most $n^{\beta}$ gates that are not $\delta$-close to an explicit constant. Let us denote this set of gates by $G_\sigma$.
\end{itemize}
\end{itemize}

In~\cite{KL}, such a subroutine $\mathcal{A}_2(C,d,n,M)$ was designed. Specifically, they proved the following theorem.

\begin{theorem}[Kabanets and Lu~\cite{KL}]
\label{thm:KL}
There is a zero-error randomized algorithm $\mathcal{A}_2(C,d,n,M)$ that runs in time $\poly(n,M)\cdot O(2^{n-n^{1-2\beta}})$ and outputs a decision tree as described above with probability at least $1-1/2^{10n}$ (and outputs $?$ otherwise). Moreover, given a good $C_\sigma$, there is a deterministic algorithm that runs in time $\poly(n,M)$ which computes $B_\sigma$ and $G_\sigma$. 
\end{theorem}
\begin{remark}
In~\cite{KL}, it is easy to see that the probability of outputting $?$ is at most $1/2$. To bring down this probability to $1/2^{10n}$, we run their procedure in parallel $10n$ times, and output the first tree that is output by the algorithm. The probability that no such tree is output is $1/2^{10n}.$ 
\end{remark}

\begin{remark}
In designing the above subroutine in~\cite{KL}, they consider a more general class of polynomially sparse-PTF circuits (i.e. each gate computes a PTF with polynomially many monomials) as opposed to the $k$-PTF circuits we consider here. Under this weaker assumption, they get that $\delta = \exp(-n^{\Omega(\beta^3)})$. However, by redoing their analysis for degree $k$-PTFs, it is easy to see that $\delta$ could be set to $\exp(-n^{\beta/B\cdot k^2})$ for some constant $B$. Under this setting of $\delta,$ we get exactly the same guarantees. In this sense, the above theorem statement is a slight restatement of~\cite[Theorem 3.11]{KL}.
\end{remark}

\subsection{Enumerating the minority assignments}
\label{sec:minority}

We now design an algorithm $\mathcal{A}_3(m,\ell,\delta, g_1, \ldots, g_\ell)$, which has the following behaviour. 
\begin{itemize}
\item[] {\bf Input:} parameters $m\leq n, \ell, \delta$ such that $\delta \in \left[\exp(-m^{1/10(k+1)}) ,1\right]$, $\ell \leq m^2$, $k$-PTFs $g_1, g_2, \ldots, g_\ell$ specified by polynomials $P_1, \ldots, P_\ell$ on $m$ variables ($x_1,\ldots,x_m$) each of weight at most $M$ and which are $\delta$-close to $-1$.

\item[] {\bf Oracle access to:} $\mathcal{A}_1$. 

\item[] {\bf Output:} $a \in \{-1,1\}^m$ such that $\exists i \in [\ell]$ for which $P_i(a)>0$. 

\end{itemize}

\begin{lemma}
\label{lem:min_count}
There is a deterministic algorithm $\mathcal{A}_3(m,\ell,\delta, g_1, \ldots, g_\ell)$ as specified above that runs in time $\poly(m,M)\cdot \sqrt{\delta}\cdot 2^{m}$. 
\end{lemma}
\begin{proof}
We start with the description of the algorithm.

\noindent
\paragraph{$\mathbf{\mc{A}_3(m,\ell,\delta, g_1, \ldots, g_\ell)}$}

\begin{enumerate}
\item Set $q = \frac{1}{2}\log \frac{1}{\delta} \leq \frac{m}{2}$ and let $\mc{N} = \emptyset$. ($\mc{N}$ will eventually be the collection of minority assignments i.e. all $a\in \{-1,1\}^m$ such that $\exists i\in [\ell]$ for which $P_i(a)>0$.)
\item For each setting $\rho\in \{-1,1\}^{m-q}$ to the variables $x_{q+1},\ldots,x_m$, do the following:
\begin{enumerate}
\item Construct the restricted polynomials $P_{1,\rho},\ldots,P_{\ell,\rho}$. Let $g_{i,\rho} = \sgn(P_{i,\rho})$ for $i\in[\ell]$.
\item Using oracle $\mc{A}_1(q,1,-g_{i,\rho})$, check for each $i\in [\ell]$ if $g_{i,\rho}$ is the constant function $-1$ by checking if the output of the oracle on the input $-g_{i,\rho}$ is zero.
\item If there is an $i\in [\ell]$ such that $g_{i,\rho}$ is not the constant function $-1$,  try all possible assignments $\chi$ to the remaining $q$ variables $x_1,\ldots,x_q$. This way, enumerate all assignments $b = (\chi ,\rho)$ to $x_1,\ldots,x_m$ for which there is an $i\in[\ell]$ such that $P_i(b)>0$. Add such an assignment to the collection $\mc{N}$.
\end{enumerate}
\item Output $\mc{N}$.
\end{enumerate}

\paragraph{Correctness.}If $a\in\{-1,1\}^m$ is a minority assignment (i.e. $\exists i_0\in[\ell]$ so that $P_{i_0}(a)<0$) and if $a = (\chi,\rho)$ where $\rho$ is an assignment to the last $m-q$ variables, and $\chi$ to the first $q$, $a$ will get added to $\mc{N}$ in the loop of step 2 corresponding to $\rho$ and that of $\chi$ in step 2c, because of $i_0$ being a witness. Conversely, observe that we only add to the collection $\mc{N}$ when we encounter a minority assignment.

\paragraph{Running time.} For each setting $\rho\in \{-1,1\}^{m-q}$ to the variables $x_{q+1},\ldots,x_m$, step 2a takes $\poly(m,M)$ time and step 2b takes $O(\ell) = O(m^2)$ time and so combined, they take only $\poly(m,M)$ time. Let $\mc{T}$ be the set consisting of all assignments $\rho$ to the last $m-q$ variables such that the algorithm enters the loop described in step 2c i.e. 
\[\mc{T} = \{\rho\in \{-1,1\}^{m-q}|\exists i \in [\ell]:g_{i,\rho} \text{  is not the constant function}-1\}\]
and let $\mc{T}^c$ denote its complement. Also note that for a $\rho \in \mc{T}$, enumeration of minority assignments in step 2c takes $2^q \cdot \ell \cdot \poly(m,M)$ time. Therefore, we can bound the total running time by
\[\poly(m,M)(2^q\cdot |\mc{T}| + |\mc{T}^c|).\]
Next, we claim that the size of $\mc{T}$ is small:
\begin{claim}
$|\mc{T}|\leq \ell \cdot \sqrt{\delta}\cdot 2^{m-q}$.
\end{claim}
\begin{proof} We define for $i\in [\ell]$, $\mc{T}_i =\{\rho\in \{-1,1\}^{m-q}|g_{i,\rho} \text{  is not the constant function}-1\}$.
By the union bound, it is sufficient to show that $|\mc{T}_i| \leq \sqrt{\delta}\cdot 2^{m-q}$ for a fixed $i\in[\ell]$. Let $D_m$ denote the uniform distribution on $\{-1,1\}^m$ i.e. on all possible assignments to the variables $x_1,\ldots,x_m$. Then from the definition of $\delta$-closeness, we know
\[\Pr_{a\sim D_m}[g_i(a) = 1] \leq \delta.\]
Writing LHS in the following way, we have
\[\avg{\rho\sim D_{m-q}}{\Pr_{\chi\sim D_q}[g_{i,\rho} (\chi) = 1]} \leq \delta\]
where $D_{m-q}$ and $D_q$ denote uniform distributions on assignments to the last $m-q$ variables and the first $q$ variables respectively. By Markov's inequality,
\[\Pr_{\rho\sim D_{m-q}}[\Pr_{\chi\sim D_q}[g_{i,\rho} (\chi) = 1] \geq \sqrt{\delta}] \leq \sqrt{\delta}\]
Consider a $\rho$ for which this event does not occur i.e. for which
$\Pr_{\chi\sim D_q}[g_{i,\rho} (\chi) = 1] < \sqrt{\delta}$.
For such a $\rho$, $g_{i,\rho}$ has only $2^q = 1/\sqrt{\delta}$ many inputs and therefore, $g_{i,\rho}$ must be the constant function $-1$. Thus, we conclude that 
\[\Pr_{\rho\sim D_{m-q}}[g_{i,\rho} \text{ is not the constant function} -1] \leq \sqrt{\delta}\]
or in other words, $|\mc{T}_i| \leq \sqrt{\delta}\cdot 2^{m-q}$.
\end{proof}

Finally, by using the trivial bound $|\mc{T}^c| \leq 2^{m-q}$ and the above claim, we obtain a total running time of $\poly(m,M)\cdot\sqrt{\delta}\cdot 2^m$ and this concludes the proof of the lemma.
\end{proof}

\subsection{\#SAT for AND of $k$-PTFs}
\label{sec:base-case}

We design an algorithm $\mathcal{A}_4(n,M,g_1,\ldots,g_\tau)$ with the following functionality. 

\begin{itemize}
    \item[] \textbf{Input:} A set of $k$-PTFs $g_1, \ldots, g_\tau$ specified by polynomials $P_1, \ldots, P_\tau$ on $n$ variables such that $w(p_i) \leq M$ for each $i \in [\tau]$ and $\sum_{i \in [\tau]} \text{fan-in}(g_i) \leq n^{1+\varepsilon_1}.$
    \item[] \textbf{Output:} \#$\{a \in \{-1,1\}^n \mid \forall i \in [\tau], P_i(a)<0\}$. 
\end{itemize}

\subsubsection{The details of the algorithm}

\noindent
\paragraph{\(\mathbf{\mc{A}_4(n,M,g_1,\ldots,g_\tau)}\)} 
\begin{enumerate}
\item Let \(m = n^\alpha\) for \(\alpha = \frac{\zeta\varepsilon_1}{2(k+1)}\). Let $C$ denote the AND of $g_1, \ldots, g_\tau$. 

\item Run $\mathcal{A}_2(C,2,n,M)$ to obtain the decision tree \(\Tdt\). Initialize $N$ to $0$. 

\item For each leaf \(\sigma\) of \(\Tdt\), do the following:
\begin{enumerate}[label=(\Alph*)]
    \item \label{step:a4_1} If \(C_\sigma\) is not \emph{good}, count the number of satisfying assignments for \(C_\sigma\) by brute-force and add to \(N\).
    \item If \(C_\sigma\) is \emph{good}, do the following:
        \begin{enumerate}[label=(\roman*)]
            \item \(C_\sigma\) is now an AND of PTFs in \(B_\sigma\) and \(G_\sigma\), over \(n' = n^{1 - 2\beta_1}\) variables, where all PTFs in \(B_\sigma\) are \(\delta\)-close to an explicit constant, where $\delta = \exp(-n^{\beta_1/B\cdot k^2})$. Moreover, \(\abs{B_\sigma} \leq n, \abs{G_\sigma} \leq n^{\beta_1}\). 
Let $B_\sigma=\{h_1, \ldots, h_\ell\}$ be specified by $Q_1, \ldots, Q_\ell$. 

                Suppose for $i \in [\ell]$, $h_i$ is close to $a_i \in \{-1,1\}$. Then let $Q'_i = -a_i \cdot Q_i$ and $h_i' = \sgn(Q'_i)$. Let $B'_\sigma = \{Q'_1, \ldots, Q'_\ell\}$. 
            \item For each restriction \(\rho: \{x_{m+1}, \ldots, x_{n'}\} \to \{-1, 1\}\), do the following:
                \begin{enumerate}[label=(\alph*)]
                    \item \label{step:a4_3} Check if there exists \(h' \in B'_\sigma\) such that \(h'_\rho\) is not the constant function $-1$ using \(\mc{A}_1(m,1,h'_\rho)\).
                    \item \label{step:a4_4} If such an \(h' \in B'_{\sigma}\) exists, then count the number of satisfying assignments for \(C_{\sigma\rho}\) by brute-force and add to \(N\).
                    \item \label{step:a4_5} If the above does not hold, we have established that for each $h_i \in B_\sigma$, \(h_{i,\rho}\) is the constant function $a_i$. If $\exists i \in [\ell]$ such that $a_i =1$, it means $C_{\sigma\rho}$ is also a constant $1$ . Then simply halt. Else set each $h_i$ to $a_i$.

                        Thus, \(C_{\sigma\rho}\) has been reduced to an AND of \(n^{\beta_1}\) many PTFs over \(m\) variables. Call this set $G'_{\sigma\rho}$, use \(\mc{A}_1(m,n^{\beta_1},G'_{\sigma\rho})\) to calculate the number of satisfying assignments and add the output to \(N\).
                \end{enumerate}
        \end{enumerate}
  \end{enumerate}
\item Finally, output $N$.
\end{enumerate}
\subsubsection{The correctness argument and running time analysis}

\begin{lemma}
    \label{lem:base_case}
    \(\mc{A}_4\) is a zero-error randomized algorithm that counts the number of satisfying assignments correctly. Further, \(\mc{A}_4\) runs in time \(\poly(n, M)\cdot 2^{n - n^\alpha}\) and outputs the right answer with probability at least $1/2$ (and outputs ? otherwise). 
\end{lemma}

\begin{proof}
\textbf{Correctness.} For a leaf \(\sigma\) of \(\Tdt\), when \(C_\sigma\) is not \emph{good}, we simply use brute-force, which is guaranteed to be correct. Otherwise, 
            \begin{itemize}
                \item If \(h'_\rho \) not the constant function $-1$ for some \(h' \in B'_\sigma\), then we again use brute-force, which is guaranteed to work correctly.
                \item Otherwise, for each \(h' \in B'_\sigma\), $h'_\rho$ is the constant function \(-1\). Here we only need to consider the satisfying assignments for the gates in \(G_{\sigma\rho}\). For this we use \(\mc{A}_1\), that works correctly by assumption.
    \end{itemize}

    Further, we need to ensure that the parameters that we call $\mc{A}_1$ on, are valid. 
To see this, observe that $m = n^\alpha \leq n^{1/(2(k+1))}$ because of the setting of $\alpha$ and further, we have $n^{\beta_1} \leq n^{0.1}$.

Finally, the claim about the error probability follows from the error probability of $\mc{A}_2$ (Theorem~\ref{thm:KL}).

\paragraph{Running Time.} The time taken for constructing \(\Tdt\) is \(O^*(2^{n - n^{1 - 2\beta_1}})\), by Theorem~\ref{thm:KL}. For a leaf $\sigma$ of $\Tdt$,  we know that step~\ref{step:a4_1} is executed with probability at most \(2^{-n^{\varepsilon_1}}\). The total time for running step~\ref{step:a4_1} is thus \(O^*(2^{n - n^{\varepsilon_1}})\). 

We know that the oracle \(\mc{A}_1\) answers calls in $\poly(n,M)$ time.  Hence, the total time for running step~\ref{step:a4_3} is \(O^*(2^{n - n^\alpha})\). Next, note that if step~\ref{step:a4_4} is executed, then all PTFs in \(B_\sigma\) are \(\delta\)-close to \(-1\). So, the number of times it runs is at most \(\delta\cdot 2^{n'}\). Therefore, the total time for running step~\ref{step:a4_4} is \(O^*(2^{n + n^\alpha - n^{\beta_1/Bk^2}})\). Similar to the analysis of step~\ref{step:a4_3}, the total time for running  step~\ref{step:a4_5} is also \(O^*(2^{n - n^\alpha})\).

    Summing them up, we conclude that total running time is \(O^*(2^{n - n^\alpha})\), as due to our choice of various parameters, $n - n^\alpha$ is the dominating power of \(2\). This completes the proof.
\end{proof}

\subsection{\#SAT for larger depth $k$-PTF circuits}
Let $C$ be a $k$-PTF circuit of depth $d\geq 1$ on $n$ variables and let $\mathcal{P}$ be a set of $k$-PTFs $g_1, \ldots, g_\tau$, which are specified by $n$-variate polynomials $P_1, \ldots, P_\tau$. Let \#SAT$(C,\mathcal{P})$ denote \#$\{a \in \{-1,1\}^n \mid C(a)<0 \text{ and } \forall i \in [\tau], P_i(a)<0 \}$. We now specify our depth-reduction algorithm $\mathcal{A}_5(n,d,M,n^{1+\varepsilon_d}, C, \mathcal{P})$. 

\begin{itemize}
\item[] {\bf Input:} $(C, \mathcal{P})$ as follows:
\begin{itemize}
\item $k$-PTF circuit $C$ with parameters $(n,n^{1+\varepsilon_d},d,M)$.
\item a set $\mathcal{P}$ of $k$-PTFs $g_1, \ldots, g_\tau$ on $n$ variables, which are specified by polynomials $P_1, \ldots, P_\tau$ such that $\sum_{i=1}^\tau \text{fan-in}(g_i) \leq n^{1+\varepsilon_{d}}$ and for each $i \in [\tau]$, $w(P_i)\leq M$. 
\end{itemize}
\item[] {\bf Oracle access to:} $\mathcal{A}_1, \mathcal{A}_4$. 

\item[] {\bf Output:} \#SAT$(C,\mathcal{P})$.

\end{itemize}

We start by describing the algorithm. 
\subsubsection{The details of the algorithm.}

Let \texttt{count} be a global counter initialized to $0$ before the execution of the algorithm. 
\paragraph{$\mathbf{\mathcal{A}_5(n,d,M,n^{1+\varepsilon_d},C,\mathcal{P})}$} 
\begin{enumerate}
\item \label{step:a5_0} If $d=1$, output $\mc{A}_4(n,M,\{C\}\cup \mc{P})$ and halt.
    \item \label{step:a5_1} Run $\mathcal{A}_2(C,d,n,M)$, which gives us a $\Tdt$. (If not, output ?.)
    \item For each leaf $\sigma \in \{-1,1\}^{n-n^{1-2\beta_{d}}}$ of $\Tdt$.  
        \begin{enumerate}
        \item For each $i \in [\tau]$ compute $P_{i,\sigma}$, the polynomial obtained by substituting $\sigma$ in its variables. Let $\mathcal{P}_\sigma = \{P_{1,\sigma}, \ldots, P_{\tau, \sigma}\}$. 
            \item \label{step:a5_2} Obtain $C_\sigma$. If $C_\sigma$ is not a good circuit, then brute-force to find the number of satisfying assignments of $(C_\sigma,\mc{P}_\sigma)$, say $N_\sigma$, and set $\texttt{count} = \texttt{count} + N_\sigma$.  
            \item If $C_\sigma$ is good then obtain $B_\sigma$ and $G_\sigma$. 
            \item Let $B_\sigma=\{h_1, \ldots, h_\ell\}$ be specified by $Q_1, \ldots, Q_\ell$. We know that each $h\in B_\sigma$ is $\delta$-close to an explicit constant, for $\delta = \exp(-n^{\beta_d/Bk^2}).$ 

                Suppose for $i \in [\ell]$, $h_i$ is close to $a_i \in \{-1,1\}$. Then let $Q'_i = -a_i \cdot Q_i$ and $h_i' = \sgn(Q'_i)$. Let $B'_\sigma = \{Q'_1, \ldots, Q'_\ell\}$. 

            \item \label{step:a5_3} Run $\mathcal{A}_3(n^{1-2\beta_{d}},\ell, \delta,h_1', \ldots, h_\ell')$ to obtain the set $\mc{N}_\sigma$  of all the minority assignments of $B_\sigma$. (Note that this uses oracle access to $\mathcal{A}_1$.)
                \begin{itemize}
                    \item[] for each $a\in \mc{N}_\sigma$, if ($(C(a)<0)$ AND $(\forall i \in [\ell]$, $P_i(a)<0)$), then $\texttt{count} = \texttt{count}+1$. 
                \end{itemize}

            \item \label{step:a5_3.5} Let $G_\sigma = \{f_1, \ldots, f_t\}$ be specified by polynomials $R_1, \ldots, R_t$. We know that $t \leq n^{\beta_{d}}$. 

                For each $b \in \{-1,1\}^t$, 
                \begin{enumerate}
                    \item Let $R_i' = -b_i \cdot R_i$ for $i \in [t]$. Let $G'_{\sigma,b} = \{R'_1, \ldots, R'_t\}$. 
                    \item Let $C_{\sigma,b}$ be the circuit obtained from $C_\sigma$ by replacing each $h_i$ by $a_i$ $1 \leq i \leq \ell$ and each $f_j$ by $b_j$ for $1\leq j \leq t$. 
                    \item $\mathcal{P}_{\sigma,b} = \mathcal{P}_\sigma \cup B'_\sigma \cup G'_{\sigma,b}$. 
                    \item \label{step:a5_4} If $d > 2$ then run $\mathcal{A}_5(n^{1-2\beta_d}, d-1, M, n^{1 + \varepsilon_d}, C_{\sigma, b}, \mathcal{P}_{\sigma,b} )$ $n_1=10n$ times and let $N_\sigma$ be the output of the first run that does not output ?. Set $\texttt{count} = \texttt{count} + N_\sigma$. (If all runs of $\mc{A}_5$ output ?, then output ?.)
                    \item \label{step:a5_5} If $d=2$ then run $\mathcal{A}_4(n^{1-2\beta_d}, M, C_{\sigma, b} \cup \mathcal{P}_{\sigma,b})$ $n_1=10n$ times and let $N_\sigma$ be the output of the first run that does not output ?. Set $\texttt{count} = \texttt{count} + N_\sigma.$ (If all runs of $\mc{A}_5$ output ?, then output ?.)
                \end{enumerate} 
        \end{enumerate}
    \item Output \texttt{count}.
\end{enumerate}

\subsubsection{The correctness argument and running time analysis}

\begin{lemma}
    \label{lem:ckt_sat}
    The algorithm \(\mc{A}_5\) described above is a zero-error randomized algorithm which  on input \((C, \mc{P})\) as described above, correctly \#SAT$(C,\mc{P})$. Moreover, the algorithm outputs the correct answer (and not ?) with probability at least $1/2$. Finally, \(\mc{A}_5(n,d,M,n^{1+\varepsilon_d},C,\emptyset)\) runs in time \(\poly(n,M)\cdot 2^{n - n^{\zeta \varepsilon_d/2(k+1)}}\), where parameters $\varepsilon_d, \zeta$ are as defined at the beginning of Section~\ref{sec:const-d}. 
\end{lemma}

\begin{proof}
 We argue correctness by induction on the depth $d$ of the circuit $C$.
 
 Clearly, if $d=1,$ correctness follows from the correctness of algorithm $\mc{A}_4.$ This takes care of the base case.
 
 If $d\geq 2$, we argue first that if the algorithm does not output ?, then it does output \#SAT$(C,\mc{P})$ correctly. Assume that the algorithm $\mc{A}_2$ outputs a decision tree $\Tdt$ as required (otherwise, the algorithm outputs ? and we are done). Now, it is sufficient to argue that for each $\sigma,$ the number of satisfying assignments to $(C_\sigma,\mc{P}_\sigma)$ is computed correctly (if the algorithm does not output ?).
 
 Fix any $\sigma.$ If $C_\sigma$ is not a good circuit, then the algorithm uses brute-force to compute \#SAT$(C_\sigma,\mc{P}_\sigma)$ which yields the right answer. So we may assume that $C_\sigma$ is indeed good. 
 
 Now, the satisfying assignments to $(C_\sigma,\mc{P}_\sigma)$ break into two kinds: those that are minority assignments to the set $B_\sigma$ and those that are majority assignments to $B_\sigma.$ The former set is enumerated in Step~\ref{step:a5_3} (correctly by our analysis of $\mc{A}_3$) and hence we count all these assignments in this step. 
 
 Finally, we claim that the satisfying assignments to $(C_\sigma,\mc{P}_\sigma)$ that are majority assignments of all gates in $B_\sigma$ are counted in Step~\ref{step:a5_3.5}. To see this, note that each such assignment $a\in \{-1,1\}^{n^{1-2\beta_d}}$ forces the gates in $G_\sigma$ to some values $b_1,\ldots,b_t\in \{-1,1\}$. Note that for each such $b\in \{-1,1\}^t$, these assignments are exactly the satisfying assignments of the pair $(C_{\sigma,b},\mc{P}_{\sigma,b})$ as defined in the algorithm.  In particular, the number satisfying assignments to $(C_\sigma, \mathcal{P}_\sigma)$ that are majority assignments of all gates in $B_\sigma$ can be written as 
\[
 \sum_{b\in \{-1,1\}^t} \text{\#SAT}(C_{\sigma,b},\mc{P}_{\sigma,b}).
\]
 We now want to apply the induction hypothesis to argue that all the terms in the sum are computed correctly. To do this, we need to argue that the size of $C_{\sigma,b}$ and the total fan-in of the gates in $\mc{P}_{\sigma,b}$ are bounded as required (note that the total size of $C$ remains the same, while the total fan-in of $\mc{P}$ increases by the total fan-in of the gates in $B_\sigma'\cup G_{\sigma,b}'$ which is at most $n^{1+\varepsilon_d}$). It can be checked that this boils down to the following two inequalities
 \[
 n^{(1-2\beta_d)(1+\varepsilon_{d-1})}\geq n^{1+\varepsilon_{d}}\text{  and   } n^{(1-2\beta_d)(1+\varepsilon_{d-1})}\leq 2n^{1+\varepsilon_{d}} 
 \]
 both of which are easily verified for our choice of parameters (for large enough $n$). Thus, by the induction hypothesis, all the terms in the sum are computed correctly (unless we get ?). Hence, the output of the algorithm is correct by induction.
 
 Now, we analyze the probability of error. If $d=1$, the probability of error is at most $1/2$ by the analysis of $\mc{A}_4.$ If $d > 2$, we get an error if either $\mc{A}_2$ outputs ? or there is some $\sigma$ such that the corresponding runs of $\mc{A}_5$ or $\mc{A}_4$ output ?. The probability of each is at most $1/2^{10n}$. Taking a union bound over at most $2^n$ many $\sigma,$ we see that the probability of error is at most $1/2^{\Omega(n)}\leq 1/2.$
 
 Finally, we analyze the running time. Define $\mc{T}(n,d,M)$ to be the running time of the algorithm on a pair $(C,\mc{P})$ as specified in the input description above. We need the following claim.
 
 \begin{claim}
 \label{clm:rtA5}
 $\mc{T}(n,d,M)\leq \poly(n,M)\cdot 2^{n-n^{\zeta \varepsilon_d/2(k+1)}}.$
 \end{claim}
 
 To see the above, we argue again by induction. The case $d=1$ follows from the running time of $\mc{A}_4.$ Further from the description of the algorithm, we get the following inequality for $d\geq 2.$
 \begin{equation}\label{eq:rta5}
 \mc{T}(n,d,M)\leq \poly(n,M)\cdot (2^{n-n^{1-2\beta_d}} + 2^{n-n^{\varepsilon_d}} + 2^{n-\frac{1}{2}\cdot n^{-\beta_d/(Bk^2)}} + 2^{n-n^{(1-2\beta_d)\zeta\varepsilon_{d-1}/2(k+1)}})
 \end{equation}
 The first term above accounts for the running time of $\mc{A}_2$ and all steps other Steps~ \ref{step:a5_2},\ref{step:a5_3} and \ref{step:a5_3.5}. The second term accounts for the brute force search in Step~\ref{step:a5_2} since there are only a $2^{-n^{\varepsilon_d}}$ fraction of $\sigma$ where it is performed. The third term accounts for the minority enumeration algorithm in Step~\ref{step:a5_3} (running time follows from the running time of that algorithm). The last term is the running time of Step~\ref{step:a5_3.5} and follows from the induction hypothesis. 
 
 It suffices to argue that each term in the RHS of (\ref{eq:rta5}) can be bounded by $2^{n-n^{\zeta\varepsilon_d/2(k+1)}}.$ This is an easy verification from our choice of parameters and left to the reader. This concludes the proof.
\end{proof}
\subsection{Putting it together}
\label{sec:put-together}
In this subsection, we complete the proof of Theorem~\ref{thm:ptfcktsat-intro} using the aforementioned subroutines. We also need to describe the subroutine $\mc{A}_1$, which is critical for all the other subroutines. We shall do so \emph{inside} our final algorithm for the \#SAT problem for $k$-PTF circuits, algorithm $\mc{B}$. Recall that $\mc{A}_1$ has the following specifications:

\begin{itemize}
\item[] \textbf{Input:} AND of $k$-PTFs, say $f_1, \ldots, f_s$ specified by polynomials $P_1, \ldots, P_s$ respectively, such that $s \leq n^{0.1}$ and for each $i \in [s]$, $f_i$ is defined over $n'\leq n^{1/(2(k+1))}$ variables and $w(P_i)\leq M$.
\item[]  \textbf{Output:} \#$\{a \in \{-1,1\}^{n'} \mid \forall i \in [s], f_i(a) = -1\}$. 
\end{itemize}

We are now ready to complete the proof of Theorem~\ref{thm:ptfcktsat-intro}. Suppose $C$ is the input $k$-PTF circuit with parameters $(n,n^{1+\varepsilon_d},d,M)$. On these input parameters $(C,n,n^{1+\varepsilon_d},d,k,M)$, we finally have the following algorithm for the \#SAT problem for $k$-PTF circuits:

\paragraph{$\mathbf{\mc{B}(C,n,n^{1+\varepsilon_d},d,k,M)}$}
\begin{enumerate}
\item\label{step:const} (\emph{Oracle Construction Step}) Construct the oracle $\mc{A}_1$ as follows.  Use $n_1 = 10n$ \emph{independent runs} of the algorithm from Corollary~\ref{cor:KLM}, with $\ell$ chosen to be $ n^{0.1}$ and $m$ to be $n^{1/2(k+1)}$, to construct independent random linear decision trees $T_1,\ldots,T_{n_1}$ such that on any input $\overline{w} = (\coeff_{m,k}(Q_1),\ldots,\coeff_{m,k}(Q_\ell))\in \mathbb{R}^{r\cdot \ell}$ (where $Q_i$s are polynomials of degree at most $k$ that sign-represent $k$-PTFs $g_i$, each on $m$ variables), each $T_i$ computes the number of common satisfying assignments to $g_1,\ldots,g_\ell$ with error at most $1/2$. 
\item Run $\mc{A}_5(n,d,M,n^{1+\varepsilon_d},C,\emptyset)$. For an internal call to $\mc{A}_1$, say on parameters $(n',s,f_1,\ldots,f_s)$ where $n'\leq m$ and $s\leq \ell$, do the following:
\begin{enumerate}
\item\label{step:a_1-1} Run each $T_i$ on the input $\overline{w} = (\coeff_{n',k}(P_1),\ldots,\coeff_{n',k}(P_s))\in \mathbb{R}^{r\cdot s}$. (We expand out the coefficient vectors with dummy variables so that they depend on exactly $m$ variables. Similarly, using some dummy polynomials, we can assume that there are exactly $\ell$ polynomials.)
\item\label{step:a_1-2} If some $T_i$ outputs the number of common satisfying assignments to $f_1,\ldots,f_s$, then output that. Otherwise, if all $T_i$ output $?$, then output $?$.
\end{enumerate}
\end{enumerate}

\begin{lemma}\label{lem:A_1const}
The construction of the zero-error randomized oracle $\mc{A}_1$ in the above algorithm takes $2^{O(n^{0.6})}$ time. Once constructed, the oracle $\mc{A}_1$ answers any call (with the correct parameters) in $\poly(n,M)$ time with error at most $1/2^{10n}$.
\end{lemma}

\begin{proof}
\textbf{Correctness.} It is clear from Corollary~\ref{cor:KLM} that algorithm $\mc{A}_1$ outputs either $?$ or the correct number of common satisfying assignments to $f_1,\ldots,f_s$. Further, as the $T_i$s in step~\ref{step:const} are constructed independently, it follows that with probability at least $1 - 1/2^{10n}$, the algorithm indeed outputs the number of common satisfying assignments to $f_1,\ldots,f_s$.

\textbf{Running Time.} Substituting the parameters $\ell= n^{0.1}$ and $m =  n^{1/(2(k+1))}$ in Corollary~\ref{cor:KLM}, we see that the construction of $\mc{A}_1$ (step~\ref{step:const}) takes $n_1\cdot 2^{n^{0.6}}$ time. Also, the claimed running time of answering a call follows upon observing that steps~\ref{step:a_1-1} and~\ref{step:a_1-2} combined take only $\poly(n,M)$ time to execute.
\end{proof}

With the correctness of $\mc{A}_1$ now firmly established, we finally argue the correctness and running time of algorithm $\mc{B}$.

\paragraph{Correctness.} The correctness of $\mc{B}$ follows from that of $\mc{A}_1,\mc{A}_2,\mc{A}_3,\mc{A}_4,$ and $\mc{A}_5$ (see Lemma~\ref{lem:A_1const}, Theorem~\ref{thm:KL}, Lemmas~\ref{lem:min_count},~\ref{lem:base_case}, and~\ref{lem:ckt_sat} respectively). Furthermore, if the algorithm $\mc{A}_1$ is assumed to have no error at all, then from the analysis of $\mc{A}_5,$ we see that the probability of error in $\mc{B}$ is at most $1/2$. However, as algorithm $\mc{A}_1$ is itself randomized, we still need to bound the probability that any of the calls made to $\mc{A}_1$ produce an undesirable output (i.e. an output of $?$). To this end, first note that as the running time of $\mc{A}_5$ is bounded by $2^n$, the number of calls to $\mc{A}_1$ is also bounded by $2^n$. But by Theorem~\ref{thm:KL} and Lemma~\ref{lem:A_1const}, the probability of $\mc{A}_1$ outputting $?$ is bounded by $1/2^{10n}$. Therefore, by the union bound, algorithm $\mc{B}$ correctly outputs the number of satisfying assignments to the input circuit $C$ with probability at least $1/2-1/2^{\Omega(n)}\geq 1/4$.

\paragraph{Running Time.} By Lemma~\ref{lem:ckt_sat} and~\ref{lem:A_1const}, the running time of $\mc{B}$ will be $2^{O(n^{0.6})}+\poly(n,M)\cdot 2^{n - n^{\zeta \varepsilon_d/2(k+1)}}$. Thus, the final running time is $\poly(n,M)\cdot 2^{n-S}$ where $S = n^{\zeta\varepsilon_{d}/2(k+1)}$ and where $\varepsilon_{d}>0$ is a constant depending only on $k$ and $d$. Setting $\varepsilon_{k,d} = \zeta\varepsilon_d/2(k+1)$ gives the statement of Theorem~\ref{thm:ptfcktsat-intro}.

\section{Acknowledgements}

We are grateful to Valentine Kabanets for telling us about this problem. We also thank Paul Beame, Prahladh Harsha, Valentine Kabanets, Ryan O'Donnell, Rahul Santhanam, Madhu Sudan, Avishay Tal and Ryan Williams for valuable discussions and comments.

\bibliographystyle{alpha}
\bibliography{references}{}

\end{document}